\documentclass[12pt]{article}

\usepackage{amsfonts}
\usepackage{amsmath}
\usepackage{amscd}
\usepackage{amssymb}
\usepackage{amsthm}
\usepackage{amsxtra}

\newcommand{\bbC}    {\mathbb C}

\newcommand{\bbR}    {\mathbb R}
\newcommand{\bbF}     {\mathbb F}

\newcommand{\bbZ}    {\mathbb Z}
\newcommand{\bbN}    {\mathbb N}
\newcommand{\bbI}      {\mathbb I}

\newcommand{\cC}    {{\cal C}}

\newcommand{\cE}    {{\cal E}}
\newcommand{\cF}    {{\cal F}}

\newcommand{\al}    {\alpha}
\newcommand{\be}    {\beta}
\newcommand{\de}    {\delta}

\newcommand{\eps}    {\epsilon}

\newcommand{\ga}    {\gamma}

\newcommand{\la}    {\lambda}

\newcommand{\om}    {\omega}
\newcommand{\Ga}    {\Gamma}
\newcommand{\La}    {\Lambda}

\newcommand{\A}      {\mathrm A}

\newcommand{\M}       {\mathrm M}

\newcommand{\X}       {\mathrm X}

\newcommand{\Ur}        {\mathrm U}

\newcommand{\Z}        {\mathrm Z}

\newcommand{\bA}     {\mathbf A}
\newcommand{\bB}      {\mathbf B}

\newcommand{\mbf}       {\mathbf f}

\newcommand{\bU}     {\mathbf U}

\newcommand{\bu}    {\mathbf u}

\newcommand{\bfI}     {\mathbf I}
\newcommand{\bfJ}     {\mathbf J}
\newcommand{\bj}     {\mathbf j}

\newcommand{\fD}    {\mathfrak D}

\newcommand{\fF}      {\mathfrak F}

\newcommand{\na}      {\nabla}

\newcommand{\bphi}   {\pmb{\phi}}
\newcommand{\bchi}   {\pmb{\chi}}
\newcommand{\bpsi}   {\pmb{\psi}}
\newcommand{\bxi}     {\pmb{\xi}}
\newcommand{\bet}     {\pmb{\eta}} 
\newcommand{\bzet}    {\pmb{\zeta}}

\newcommand{\bom}    {\pmb{\omega}}
\newcommand{\p}        {\partial}

\newcommand{\lan}      {\langle}
\newcommand{\ran}      {\rangle}
 
\newcommand{\Int}       {\textstyle{\int\!}}
\newcommand{\bcE}     {\pmb{\cal E}} 

\newcommand{\oF}       {\overset{\,\circ}\cF{}} 

\newcommand{\stup}[1]    {\stackrel{*}{#1}{\!}}

\DeclareMathOperator{\Hom}   {Hom}
\DeclareMathOperator{\End}   {End}

\DeclareMathOperator{\supp}  {supp\,}
\DeclareMathOperator{\im}    {Im}
\DeclareMathOperator{\ke}    {Ker}
\DeclareMathOperator{\com}   {{\scriptstyle\circ}\,}

\DeclareMathOperator{\ev}    {ev}

\DeclareMathOperator{\Div}     {Div}

\newtheorem{theorem}{Theorem}
\newtheorem{lemma}{Lemma}
\newtheorem{prop}{Proposition}

\theoremstyle{definition}

\newtheorem{defi}{Definition}
\newtheorem{ass}{Assumption}

\begin{document}

\title{On Hamiltonian operators in differential algebras}
\author
          {
             Zharinov V.V.
             \thanks{Steklov Mathematical Institute}
             \thanks{E-mail: zharinov@mi.ras.ru}
            }
\date{}
\maketitle

\begin{abstract}
Before we proposed an algebraic technics for the Hamiltonian approach 
to the evolution systems of partial differential equations, including systems 
with constraints. Here we further develop this approach and present the 
defining system of equations (suitable for the computer calculations), 
characterizing the Hamiltonian operators of the given form. 
We illustrate our technics by a simple example.
\end{abstract} 


\section{Introduction.} 
In the previous paper \cite{Z1} we proposed an algebraic technics 
(based on the book \cite{PJO}, Chapter VII) for the Hamiltonian approach 
to the evolution systems of partial differential equations, including systems 
with constraints. Here we further develop this approach and present the 
defining system of equations (suitable for the computer calculations) 
for the classification of the Hamiltonian operators of the given form. 

We work in the frames of the algebra-geometrical approach to partial differential equations 
(see, for example, \cite{PJO}, \cite{Z2}, \cite{I}). We hope that our results will be 
helpful in the studies connected with the papers \cite{Z3} -- \cite{A1}.

We use the following general notations: 
\begin{itemize} 
	\item 
		$\bbF=\bbR,\bbC$, \quad $\bbN=\{1,2,3,\dots\}\subset\bbZ_+=\{0,1,2,\dots\}$; 
	\item 
		$\M=\{1,\dots,m\}$,  \quad $m\in\bbN$;
	\item 
		$\bbI=\bbZ^\M_+=\{i=(i^1,\dots,i^m)\mid i^\mu\in\bbZ_+, \ \mu\in\M\}$.
\end{itemize} 
All linear operations are done over the number field $\bbF$. 
The summation over repeated upper and lower indices is as a rule assumed. 

\section{The differential algebra.} 
The main object in our construction is a differential algebra $(\cF,\fD_H)$, where 
(see \cite{Z1}, for details and motivations):
\begin{itemize} 
	\item 
		$\cF$ is an unital (i.e., containing the unit element) associative commutative 
		algebra; 
	\item 
		$\fD=\fD(\cF)$ is the Lie algebra of all differentiations of the algebra $\cF$,   
		it is  assumed that $\fD=\fD_V\oplus_\cF\fD_H$; 
	\item 
		$\fD_V=\fD_V(\cF)$ is the {\it vertical} subalgebra 
		of the Lie algebra $\fD$
		with the formal $\cF$-basis $\p=\{\p_a\mid a\in\bA\}$, 
		$[\p_a,\p_b]=0$, $a,b\in\bA$, $\bA$ is an index set, 
		in particular, 
		$\fD_V=\big\{X=X^a\p_a \ \big| \ X^a\in\cF\big\}$, where for any 
		$L\in\cF$ the action $\p_a L\ne0$ only for a finite number of indices $a\in\bA$;
	\item 
		$\fD_H=\fD_H(\cF)$ is the {\it horizontal} subalgebra 
		of the Lie algebra $\fD$ with the $\cF$-basis 
		$D=\{D_\mu\mid \mu\in\M\}$, $[D_\mu,D_\nu]=0$, $\mu,\nu\in\M$; 
	\item 
		$[D,\fD_V]\subset\fD_V$, i.e. $[D_\mu,V]\in\fD_V$ 
		for all $\mu\in\M$ and $V\in\fD_V$.
\end{itemize} 
Remind, that the set $\fD$ has two algebraic structures:
the structure of a Lie algebra with the commutator 
$[X,Y]=X\com Y-Y\com X$, as the Lie bracket, 
and the structure of a $\cF$-module with $(K\cdot X)L=K\cdot(XL)$, 
which are related by {\it the matching condition} 
$[X,L\cdot Y]=(XL)\cdot Y+L\cdot[X,Y]$ for all $K,L\in\cF$, 
and for all $X,Y\in\fD$. 

For any index sets $\bA,\bB$ the set  
$\cF^\bA_\bB=\big\{\bet=(\eta^a_b) \ \big| \ \eta^a_b\in\cF, \ a\in\bA, \ b\in\bB\big\}$ 
has the structure of a $\cF$-module with the component-wise multiplication, i.e.,
$K\cdot\bet=(K\cdot\eta^a_b)\in\cF^\bA_\bB$ for all $K\in\cF$, $\bet\in\cF^\bA_\bB$. 

We denote by $\oF^\bA_\bB$ the $\cF$-module of all elements 
$\bet=(\eta^a_b)\in\cF^\bA_\bB$, s.t. for any upper index $a\in\bA$ 
only a finite number of components $\eta^a_b\ne0$. 
In particular, $\oF^\bA=\cF^\bA$ for any index set $\bA$, 
while $\oF_\bB$ consists of all {\it finite} elements $\bet=(\eta_b)\in\cF_\bB$, 
and $\oF_\bB=\cF_\bB$ iff (i.e., if and only if) the index set $\bB$ is finite.  
For example, by the above assumption, $\p K=(\p_a K)\in\oF_\bA$ for any 
$K\in\cF$, in particular, the linear mapping $\p : \cF\to\oF_\bA$ is defined.

For any $\cF$-module $\cF^\bA_\bB$ and any differentiation  $X\in\fD$ the linear mapping
$X : \cF^\bA_\bB\to\cF^\bA_\bB$ is defined component-wise, 
$\bet=(\eta^a_b)\mapsto X\bet=(X\eta^a_b)$. 
In particular, {\it the Leibniz rule} takes the form: 
$X(K\cdot\bet)=(XK)\cdot\bet+K\cdot(X\bet)$ for all $X\in\fD$, $K\in\cF$, $\bet\in\cF^\bA_\bB$. 

With this notations, the commutator 
$[D_\mu,\p_a]=\Ga^b_{\mu a}\cdot\p_b$ for all  $\mu\in\M$, $a\in\bA$, 
and we {\it assume that the symbol $\Ga=(\Ga^b_{\mu a})\in\oF^\bA_{\M\bA}$}. 
Note, that here the set $\M$ is finite, while the set $\bA$ is as a rule infinite, 
and in the accordance with the above definition of the $\cF$-module 
$\oF^\bA_{\M\bA}$, for any index $b\in\bA$ only a finite number of coefficients $\Ga^b_{\mu a}\ne0$. 

\section{Horizontal differential operators.} 
We call a linear mapping $P(D) : \cF\to\cF$ (i.e., $P(D)\in\End_\bbF(\cF)$)  
{\it a horizontal differential operator}, if $P(D)=P_i\cdot D^i$, where the indices $i=(i^1,\dots,i^m)\in\bbI$,  coefficients $P=(P_i)\in\oF_\bbI$, differential monomials $D^i=(D_1)^{i^1}\com\dots\com(D_m)^{i^m}$, 
so $P(D)L=P_i\cdot D^i L$ for all $L\in\cF$. 
The set of all horizontal differential operators we denote by $\cF[D]$. 
It is an unital associative algebra with the the multiplication defined by 
{\it the composition rule}. 
In the same way for any four  index sets $\bA,\bB,\bA',\bB'$ the $\cF[D]$-module $\cF^{\bA'\bB}_{\bB'\bA}[D]$ of matrix horizontal differential operators 
\begin{equation*} 
	P(D) : \oF^\bA_\bB\to\oF^{\bA'}_{\bB'}, \quad 
	\bet=(\eta^a_b)\mapsto\bchi=(\chi^{a'}_{b'}), 
	\quad \chi^{a'}_{b'}=P^{a'b}_{b'ai}\cdot D^i\eta^a_b, 
\end{equation*}
is defined, 
where $P=\big(P^{a'b}_{b'ai}\big)\in\oF^{\bA'\bB}_{\bB'\bA\bbI}$, 
$P(D)=\big(P^{a'b}_{b'ai}\cdot D^i\big)$.

For every pair of index sets $\bA,\bB$ there is defined the natural pairing 
\begin{equation*} 
	\lan\cdot,\cdot\ran : \oF^\bA_\bB\times\oF^\bB_\bA\to\cF, \quad 
	\bet=(\eta^a_b),\bzet=(\zeta^b_a)\mapsto\lan\bet,\bzet\ran=\eta^a_b\cdot\zeta^b_a. 
\end{equation*}

For every horizontal differential operator $P(D) : \oF^\bA_\bB\to\oF^{\A'}_{\bB'}$ 
there is defined the {\it Lagrange dual} operator 
\begin{equation*} 
	\stup{P}(D) : \oF^{\bB'}_{\bA'}\to\oF^\bB_\bA, \quad 
	\bzet=(\zeta^{b'}_{a'})\mapsto\bom=(\om^b_a), \quad
	\om^b_a=(-D)^i\big(\zeta^{b'}_{a'}\cdot P^{a'b}_{b'ai}\big), 
\end{equation*}
where 
$\stup{P}=\big(\stup{P}^{ba'}_{ab'i}\big)\in\oF^{\bB\bA'}_{\bA\bB'\bbI}$, 
$\stup{P}(D)=\big(\stup{P}^{ba'}_{ab'i}\cdot D^i\big)
=\big((-D)^i\com P^{a'b}_{b'ai}\big)$, 
and the {\it Green's formula} (the integration by parts)
\begin{equation*} 
	\lan\bzet,P(D)\bet\ran-\lan\stup{P}(D)\bzet,\bet\ran=D_\mu\psi^\mu 
	\quad\text{for all}\quad \bzet\in\oF^{\bB'}_{\bA'}, \ \bet\in\oF^\bB_\bA, 
\end{equation*}
holds with some {\it current} $\psi=\big(\psi^\mu(\bzet,P(D),\bet)\big)\in\cF^\M$. 

\begin{ass} 
We assume that the differential algebra $(\cF,\fD_H)$ 
is {\it horizontally exact}, i.e. a horizontal differential operator 
$P(D)=P_i\cdot D^i=0$ (i.e., $P(D)K=0$ for all $K\in\cF$) 
iff the coefficients $P_i=0$ for all $i\in\bbI$. 
\end{ass} 
Note, that in this case, the same is also true for matrix operators. 

\section{The main ingredients}
The main components, instruments and assumptions here are 
(see \cite{Z1}, for more detail): 
\begin{itemize} 
	\item 
		the horizontal differential operator 
		\begin{equation*} 
			D : \cF\to\cF_\M, \quad K\mapsto DK =(D_\mu K);
		\end{equation*}
	\item 
		the horizontal differential operator 
		\begin{equation*}
			\Div=-\stup{D} : \cF^\M\to\cF, 
			\quad  \psi=(\psi^\mu)\mapsto\Div\psi=D_\mu\psi^\mu;
		\end{equation*}
	\item 
		the horizontal differential operator 
		\begin{equation*}
			\na : \cF^\bA\to\cF^\bA_\M, \quad \bphi=(\phi^a)\mapsto\bet=(\eta^a_\mu), 
			\quad \eta^a_\mu=D_\mu\phi^a+\Ga^a_{\mu b}\cdot\phi^b; 
		\end{equation*}
	\item 
		the Lagrange dual operator 
		\begin{equation*} 
			\stup{\na} : \oF^\M_\bA\to\oF_\bA, \quad 
			\bchi=(\chi^\mu_a)\mapsto\mbf=(f_a), \quad 
			f_a=-D_\mu\chi^\mu_a+\Ga^b_{\mu a}\chi^\mu_b; 
		\end{equation*} 
	\item 
		the {\it vertical differential operator} 
		$\p : \cF\to\oF_\bA$, $K\mapsto\p K=(\p_a K)$. 
\end{itemize}  
The operator $\na$ is correctly defined due to assumption 
$\Ga=(\Ga^b_{\mu a})\in\oF^\bA_{\M\bA})$.

\begin{prop}\label{P1} 
The compositions $\p\com\Div, \stup{\na}\com\p : \cF^\M\to\oF_\bA$, 
and the equality $\p\com\Div+\stup{\na}\com\p=0$ holds.
\end{prop} 
\begin{proof} 
Indeed, for any $\psi=(\psi^\mu)\in\cF^\M$ we have 
\begin{equation*} 
	(\p\com D_\mu)\psi^\mu+(\stup{\na}\com\p)\psi^\mu
	=\big(-[D_\mu,\p_a]+\Ga^b_{\mu a}\p_b\big)\psi^\mu=0, 
\end{equation*}
in the accordance with the definition of the symbol $\Ga$.
\end{proof}

We also need: 
\begin{itemize}
	\item 
		the linear subspace $\bcE=\ke\na$ of the linear space $\cF^\bA$; 
	\item 
		the set $\fD_E=\fD_E(\cF)=\{V\in\fD_V\mid [D_\mu,V]=0, \ \mu\in\M\}$ 
		of all {\it evolutionary} differentiations of the algebra $\cF$. 

\end{itemize} 
The set $\fD_E$ is a subalgebra of the Lie algebra $\fD_V$, because 
$[\fD_E,\fD_E]\subset\fD_E$ due to the {\it Jacobi identity} for commutators, 
but it is not a submodule of the $\cF$-module $\fD_V$ (see, e.g., \cite{Z2}). 

\begin{prop}\label{P2} 
The mapping $\ev : \bcE\to\fD_E$, $\bphi=(\phi^a)\mapsto\ev_{\bphi}=\phi^a\cdot\p_a$, 
is an isomorphism of linear spaces. 
Moreover, the structure of the Lie algebra on $\fD_E$ 
defines the isomorphic structure on $\bcE$ by the rule: 
\begin{equation*}
	[\bphi,\bpsi]=\bxi,  \quad \bphi=(\phi^a), \bpsi=(\psi^a), \bxi=(\xi^a),
	\quad \xi^a=\ev_{\bphi}\psi^a-\ev_{\bpsi}\phi^a, \quad a\in\bA.
\end{equation*} 
\end{prop} 
\begin{proof} 
The proof is done by the direct test.
\end{proof}

\begin{prop}\label{P3} 
 For every horizontal differential operator 
 $P(D)=\big(P^a_{ib}\cdot D^i\big) : \cF^\bB\to\cF^\bA$, 
for its Lagrange dual $\stup{P}(D)=\big((-D)^i\com P^a_{ib}\big) : \oF_\bA\to\oF_\bB$ 
 and for any $\bphi\in\bcE$ the following statements hold: 
 \begin{itemize} 
 	\item 
 		$[\ev_{\bphi}, P(D)]=\ev_{\bphi}P(D)=\big((\ev_{\bphi}P^a_{ib})\cdot D^i\big)$; 
	\item 
		$[\ev_{\bphi},\stup{P}(D)]=\ev_{\bphi}\stup{P}(D)
			=\big((-D)^i\com(\ev_{\bphi}P^a_{ib})\big)=[\ev_{\bphi},P(D)]^*$, 
 \end{itemize} 
 i.e., the evolutionary differentiation $\ev_{\bphi}$ acts here coefficient-wise. 
 \end{prop} 
 \begin{proof} 
 The proof is based on the characteristic property of the evolutionary differentiations. 
 \end{proof}
 
 \begin{ass}\label{ASS} 
 We assume that there exist an index set $\A$ and a horizontal differential operator 
		$\bfJ=J(D) : \cF^\A\to\cF^\bA$, $\phi=(\phi^\al)\mapsto\bphi=(\phi^a)$, s.t.,  
		\begin{itemize} 
			\item[\sffamily (a)]\quad  
				the composition $\na\com\bfJ=0$, i.e., $\im\bfJ\subset\ke\na$, 
			\item[\sffamily (b)]\quad  
			the commutator $[\ev_{\bphi},\bfJ)]=0$ for all $\bphi\in\bcE$. 
		\end{itemize} 
In more detail, $\bfJ=\big(J^a_{\al i}\cdot D^i\big)$, 
$J=(J^a_{\al i})\in\oF^\bA_{\A\bbI}$, $\phi^a=J^a_{\al i}\cdot D^i\phi^\al$. 
 \end{ass}

\begin{prop}\label{P4} 
Let the horizontal differential operator $\stup{\bfJ} : \oF_\bA\to\oF_\A$ 
be the Lagrange dual to $\bfJ : \cF^\A\to\cF^\bA$, then 
\begin{itemize} 
	\item[\sffamily (a*)] \quad  
		the composition $\stup{\bfJ}\com\stup{\na}=0$, 
	\item[\sffamily (b*)] \quad 
		the commutator $[\ev_{\bphi},\stup{\bfJ}]=0$ for all $\bphi\in\bcE$.
\end{itemize}
\end{prop} 
\begin{proof} 
The proof is done by the direct test.
\end{proof}

\section{The Lie-Poisson structures.} 
We shall use the notation: 
\begin{itemize} 
	\item 
		$\fF=\cF\big/\im\Div=\big\{\Int K=K+\Div\cF^\M \ \big| \ K\in\cF\big\}$; 
	\item 
		$\cE=\cF^\A$, $\cE^*=\Hom_\cF(\cE;\cF)=\oF_\A$.
\end{itemize}
Here, $\fF$ is the linear space of all {\it functionals} over $\cF$, 
$\cE$ is a linear space, and  $\cE^*$ is its dual. 
The natural projection $\cF\to\fF$ is defined by the rule: 
$K\mapsto \Int K=K+\Div\cF^\M$. 

\begin{ass} 
We assume that the differential algebra $(\cF,\fD_H)$ is of the 
{\it du Bois-Rey\-mond type}, i.e., it has the following property: 
for a given $K\in\cF$ the equality $\Int K\cdot L=0$ is valid 
for all $L\in\cF$ iff  $K=0$. 
\end{ass}

By Propositions \ref{P1} and  \ref{P4}, $\p : \im\Div\to\im\stup{\na}$ and $\stup{\bfJ}\com\stup{\na}=0$, 
hence we have the sequence of linear spaces: 
\begin{equation*} 
	\cF\big/\im\Div\xrightarrow{ \ \ \p \ \ }\oF_\bA\big/\im\stup{\na}
	\xrightarrow{ \ \ \ \stup{\bfJ} \ \ \ }\oF_\A.
\end{equation*}
In particular, the linear mapping (the {\it variational derivative}) 
\begin{itemize} 
	\item 
		$\de=\stup{\bfJ}\com\p : \fF\to\cE^*$, 
		\quad $\Int K\mapsto\de K=\stup{\bfJ}(\p K)$, 
\end{itemize}
is defined.

\begin{defi} The {\it Lie-Poisson structure over the differential algebra} $(\cF,\fD_H)$
is a bilinear mapping ({\it the Lie-Poisson bracket})
\begin{equation*} 
	\{\cdot,\cdot\} : \fF\times\fF\to\fF, \quad \Int K, \Int L\mapsto\big\{\Int K,\Int L\big\}, 
\end{equation*}
with the properties: 
\begin{itemize} 
	\item
		$\big\{\Int K,\Int L\big\}+\big\{\Int L,\Int K\big \}=0$; 
		\hfill ({\it skew-symmetry})
	\item
		$\bfJ\bfI\big(\Int K,\Int L,\Int M\big)
			=\big\{\Int K,\big\{\Int L,\Int M\big\}\big\} +\text{c.p.}=0$; 
			\hfill ({\it Jacobi identity})
\end{itemize}
where the abbreviation ``c.p.'' stands for the cyclic permutation 
of arguments $\Int K,\Int L,\Int M\in\fF$. 
In this case, the pair $(\fF,\{\cdot,\cdot\})$ is a Lie algebra. 
\end{defi}

\begin{defi} 
We define the bracket $\{\cdot,\cdot\}$ by the rule:
\begin{equation*} 
	\{\Int K,\Int L\}=\Int\lan\de K,\La(D)\de L\ran 
	\quad\text{for all}\quad \Int K,\Int L\in\fF,
\end{equation*}
where a horizontal differential operator 
\begin{equation*} 
	\La(D) :\cE^*\to\cE, \quad f=(f_\al)\mapsto\phi=(\phi^\al), 
	\quad \phi^\al=\La^{\al\be}_i\cdot D^if_\be, \quad \al\in\A, 
\end{equation*} 
the set of the coefficients $\La=(\La^{\al\be}_i)\in\oF^{\A\A}_\bbI$. 
{\it Below we always assume that the operator $\La(D)$ is the Lagrange skew-adjoint}, 
i.e.,  $\La(D)+\!\stup{\La}(D)=0$. 
\end{defi}
In this case, the bracket $\{\cdot,\cdot\}$ is skew-symmetric, 
and the problem left is to find a possibly simple and effective test for the operator 
$\La(D)$ to be {\it Hamiltonian}, i.e., to ensure the Jacobi identity to be valid. 

In the paper \cite{Z1} we have proved the following algebraic analog of the 
Theorem 7.8 from the book \cite{PJO}. 
\begin{lemma}\label{L1} 
The Jacoby identity has the following representation: 
\begin{equation*} 
	\bfJ\bfI(\Int K,\Int L,\Int M)
		=\Int\lan\de K,[\ev_{\bphi(L)},\La(D)]\de M\ran+\text{\rm c.p.}=0, 
\end{equation*} 
where $\bphi(L)=(\bfJ\com\La(D)\com\de)L$, 
$\La(D)=(\La^{\al\be}_i\cdot D^i)$ is a horizontal differential skew-adjoint operator. 
\end{lemma}
Here we refine this result to a form more convenient for actual calculations. 

\begin{defi} 
We shall say that the Jacobi identity holds {\it in the strengthened form} if 
\begin{equation*} 
	\Int\lan f,[\ev_{\bphi(g)},\La(D)]h\ran+\text{\rm c.p.}=0, 
\end{equation*} 
for all $f,g,h\in\cE^*$, where $\bphi(g)=(\bfJ\com\La(D))g$. 
\end{defi} 

\begin{lemma}\label{L2} 
In the conditions of the Lemma \ref{L1} the Jacoby identity 
in the strengthened form can be written as 
\begin{equation*} 
	R^{\al\be\ga}_{kl}\cdot D^k g_\be\cdot D^l h_\ga
		+(-D)^l\big(R^{\be\ga\al}_{kl}\cdot g_\be\cdot D^k h_\ga\big)
		+(-D)^k\big(R^{\ga\al\be}_{kl}\cdot D^l g_\be\cdot h_\ga\big)=0
\end{equation*}
for all $g,h\in\cE^*$, $\al,\be,\ga\in\A$, $k,l\in\bbI$, 
where the coefficients 
\begin{equation*}
	R^{\al\be\ga}_{kl}=R^{\al\be\ga}_{kl}(\La)=\tbinom{i}{r}
	J^a_{\eps i}\cdot D^{i-r}\La^{\eps\be}_{k-r}\cdot\p_a\La^{\al\ga}_l. 
\end{equation*}
\end{lemma}
\begin{proof} 
We need to refine the expression 
\begin{equation*}
\Int\lan f,[\ev_{\bphi(g)},\La(D)]h\ran
=\Int f_\al\cdot[\ev_{\bphi(g)},\La^{\al\ga}_l\cdot D^l]h_\ga, 
\end{equation*}
where $f,g,h\in\cE^*$. 
Here, $\bphi(g)=(\bfJ\com\La(D))g$, i.e., 
\begin{align*}
	\phi^a(g)
	&=J^a_{\eps i}\cdot D^i(\La^{\eps\be}_k\cdot D^k g_\be) 
		=J^a_{\eps i}\cdot\tbinom{i}{r}\cdot D^{i-r}\La^{\eps\be}_k\cdot 
		D^{k+r}g_\be \\
	&=\tbinom{i}{r}\cdot J^a_{\eps i}\cdot D^{i-r}\La^{\eps\be}_{k-r}\cdot 
		D^k g_\be.
\end{align*}
Hence, by Proposition \ref{P3}, 
\begin{align*} 
	\Int f_\al\cdot[\ev_{\bphi(g)},\La^{\al\ga}_l\cdot D^l]h_\ga 
	&=\Int f_\al\cdot\tbinom{i}{r}\cdot J^a_{\eps i}\cdot D^{i-r}\La^{\eps\be}_{k-r}
		\cdot D^k g_\be\cdot\p_a\La^{\al\ga}_l\cdot D^l h_\ga \\
	&=\Int R^{\al\be\ga}_{kl}\cdot f_\al\cdot D^k g_\be\cdot D^l h_\ga.
\end{align*}
Thus, 
\begin{align*}  
	&\Int\lan\de K,[\ev_{\bphi(L)},\La(D)]\de M\ran+\text{\rm c.p.} \\
	&=\Int R^{\al\be\ga}_{kl}\cdot\big(f_\al\cdot D^k g_\be\cdot D^l h_\ga
		+g_\al\cdot D^k h_\be\cdot D^l f_\ga+h_\al\cdot D^k f_\be\cdot D^l g_\ga\big) \\
	&=\Int\big(R^{\al\be\ga}_{kl}\!\cdot\! f_\al\!\cdot\! D^k g_\be\!\cdot\! D^l h_\ga
		+R^{\be\ga\al}_{kl}\!\cdot\! D^l f_\al\!\cdot\! g_\be\!\cdot\! D^k h_\ga
		+R^{\ga\al\be}_{kl}\!\cdot\! D^k f_\al\!\cdot\! D^l g_\be\!\cdot\! h_\ga\big) \\
	&=\Int f_\al\!\cdot\!\big(R^{\al\be\ga}_{kl}\!\cdot\! D^k g_\be\!\cdot\! D^l h_\ga
		\!+\!(-D)^l(R^{\be\ga\al}_{kl}\!\cdot\! g_\be\!\cdot\! D^k h_\ga)
		\!+\!(-D)^k(R^{\ga\al\be}_{kl}\!\cdot\! D^l g_\be\!\cdot\! h_\ga)\big), 
\end{align*}
where we twice used the Green's formula of the integration by parts. 
Clear, the last representation and the assumed du Bois-Reymond property of 
the differential algebra $(\cF,\fD_H)$ imply our claim. 
\end{proof}

\begin{theorem}\label{T1} 
In the conditions of the Lemma \ref{L2} the Jacoby identity 
in the strengthened form is reduced to the system: 
$Q=Q(\La)=(Q^{\al\be\ga}_{kl})=0$, or in more detail,
\begin{equation*}
	Q^{\al\be\ga}_{kl}=0 \quad\text{for all}\quad  \al,\be,\ga\in\A, \ \ k,l\in\bbI, 
\end{equation*}
where 
\begin{equation*} 
	Q^{\al\be\ga}_{kl}=Q^{\al\be\ga}_{kl}(\La)=R^{\al\be\ga}_{kl}
	+(-1)^j\tbinom{j}{p,k,l-i}D^p R^{\be\ga\al}_{ij}
	+(-1)^i\tbinom{i}{p,k-j,l}D^p R^{\ga\al\be}_{ij}, 
\end{equation*}
and $\tbinom{j}{p,k,l-i},\tbinom{i}{p,k-j,l}$ are the trinomial coefficients
\footnote{Remind, the polynomial coefficients are 
$\binom{i}{j_1\dots j_p}=\frac{i!}{j_1!\dots j_p!}$, $i=j_1+\dots j_p$. 
In particular, the binomial coefficient $\binom{i}{j}=\binom{i}{j,i-j}$.}. 
\end{theorem}
\begin{proof} 
Indeed, here 
\begin{align*} 
	&R^{\al\be\ga}_{kl}\cdot D^k g_\be\cdot D^l h_\ga
		+(-D)^l\big(R^{\be\ga\al}_{kl}\cdot g_\be\cdot D^k h_\ga\big)
		+(-D)^k\big(R^{\ga\al\be}_{kl}\cdot D^l g_\be\cdot h_\ga\big) \\
	&=R^{\al\be\ga}_{kl}\cdot D^k g_\be\cdot D^l h_\ga
	   +(-1)^j\tbinom{j}{pqr}D^p R^{\be\ga\al}_{kj}\cdot D^q g_\be\cdot d^{r+k}h_\ga \\
	&+(-1)^i\tbinom{i}{pqr}D^p R^{\ga\al\be}_{il}\cdot D^{q+l}g_\be\cdot D^r h_\ga \\ 
	&=R^{\al\be\ga}_{kl}\cdot D^k g_\be\cdot D^l h_\ga 
	  +(-1)^j\tbinom{j}{p,k,l-i}D^p R^{\be\ga\al}_{ij}\cdot D^k g_\be\cdot D^l h_\ga \\
      &+(-1)^i \tbinom{i}{p,k-j,l}D^p R^{\ga\al\be}_{ij}\cdot D^k g_\be\cdot D^l h_\ga 
        =Q^{\al\be\ga}_{kl}\cdot D^k g_\be\cdot D^l h_\ga. 
\end{align*} 
To complete the proof it is enough to use the assumed horizontal exactness 
of the differential algebra $(\cF,\fD_H)$. 
\end{proof} 

\section{Evolution without constraints.} 
Here (see \cite{Z1}, for more detail), 
the algebra $\cF=\cC^\infty_{fin}(\X\bU)$ is the algebra of all smooth functions 
on the infinite dimensional space $\X\bU=\X\times\bU$, 
depending on a finite number of the arguments $x^\mu,u^\al_i$, where 
\begin{itemize} 
	\item 
		$\X=\bbR^\M=\{x=(x^\mu)\mid x^\mu\in\bbR, \ \mu\in\M\}$ 
		is the space of independent variables; 
	\item 
		$\Ur=\bbR^\A=\{u=(u^\al)\mid u^\al\in\bbR, \ \al\in\A\}$ 
		is the space of dependent variables, $\A$ is a finite index set; 
	\item  
		$\bU=\bbR^\A_\bbI=\{\bu=(u^\al_i)\mid u^\al_i\in\bbR, \ \al\in\A, \ i\in\bbI\}$ 
		is the space of differential variables, $u^\al=u^\al_0$. 
\end{itemize}
The Lie subalgebra $\fD_H$ has the $\cF$-basis $D=\{D_\mu\mid \mu\in\M\}$, 
where the {\it total partial derivatives} 
\begin{equation*} 
	D_\mu=\p_{x^\mu}+u^\al_{i+(\mu)}\p_{u^\al_i}, 
	\quad i+(\mu)=(i^1,\dots,i^\mu+1,\dots,i^m). 
\end{equation*}
The Lie subalgebra $\fD_V$ has the $\cF$-basis $\{\p_{u^\al_i}\mid \al\in\A, \ i\in\bbI\}$, 
thus, here, $a=\tbinom{\al}{i}$, $\bA=\tbinom{\A}{\bbI}$, 
and the commutator 
\begin{equation*}
	[D_\mu,\p_{u^\al_i}]=-\p_{u^\al_{i-(\mu)}}, 
	\quad\text{i.e.,}\quad \Ga^{i\be}_{\mu\al j}=-\de^\be_\al\de^i_{j+(\mu)}, 
	\quad \al,\be\in\A, \ i,j\in\bbI, \ \mu\in\M.
\end{equation*} 
Note, the pair $\tbinom{\be}{j}$ is the single upper index, 
while $\tbinom{i}{\al}$ is the single lower index.  

Further, in this situation, 
\begin{itemize} 
	\item 
		$\na : \cF^\A_\bbI\to\cF^\A_{\M\bbI}, \quad  
		\bphi=(\phi^\al_i)\mapsto\bet=(\eta^\al_{\mu i}), \quad 
		\eta^\al_{\mu i}=D_\mu\phi^\al_i-\phi^\al_{i+(\mu)}$; 
	\item 
		$\bj : \cF^\A\to\cF^\A_\bbI, \quad \phi=(\phi^\al)\mapsto\bphi=(\phi^\al_i), 
		\quad \phi^\al_i=D^i\phi^\al=\de_{ik}\de^\al_\be\cdot D^k\phi^\be$; 
	\item 
		$\ke\bj=0, \quad \im\bj=\ke\na$, \quad  
		$[\ev_{\bphi},\bj]=0$ for all $\bphi\in\bcE$; 
	\item 
		$Q^{\al\be\ga}_{kl}=R^{\al\be\ga}_{kl}
			+(-1)^j\tbinom{j}{p,k,l-i}D^p R^{\be\ga\al}_{ij}
			+(-1)^i\tbinom{i}{p,k-j,l}D^p R^{\ga\al\be}_{ij}$;
	\item 
		$R^{\al\be\ga}_{kl}=\tbinom{i}{r}D^{i-r}\La^{\eps\be}_{k-r}\cdot
			\p_{u^\eps_i}\La^{\al\ga}_l$, \quad  $\al,\be,\ga\in\A$, 
			\quad  $k,l,i,r,j,p\in\bbI$.
\end{itemize}

\section{The simplest case.}
Here, 
\begin{itemize} 
	\item 
		$\X=\bbR$, \quad $\Ur=\bbR$, \quad $\bU=\bbR_\bbI$, 
		\quad $\bbI=\bbZ_+$, \quad $\cF=\cC^\infty_{fin}(\X\bU)$;
	\item  
		$D=\p_x+u_{i+1}\p_{u_i}$, \quad  $[D,\p_{u_i}]=-\p_{u_{i-1}}$, 
		\quad $i\in\bbI$.
\end{itemize}
Further in this case, the horizontal differential operator 
$\La(D)=\sum_{i=0}^s\La_i\cdot D^i$, 
where {\it the order} $s\in\bbN$, the coefficients 
$\La_i=\La_i(u_0,u_1,\dots,u_{n(i)})\in\cF$, $0\le i\le s$.  
In particular (see Lemma \ref{L2}, Theorem \ref{T1}), 
\begin{itemize} 
	\item 
		$R_{kl}=\binom{i}{r}D^{i-r}\La_{k-r}\cdot\p_{u_i}\La_l$, 
		\quad $k,l,i,r,j,p\in\bbI$; 
	\item 
		$Q_{kl}=R_{kl}+(-1)^j\binom{j}{p,k,l-i}D^p R_{ij}
			+(-1)^i\binom{i}{p,k-j,l}D^p R_{ij}$; 
	\item 
		$\supp R=\{(k,l)\in\bbI^2 \mid l\le s, \ \ k\le n(l)+s \}$; 
	\item 
		$sp\,R_{kl}=\{(r,i)\in\bbI^2\mid\max\{0,k-s\}\le r\le k, \ r\le i\le n(l)\}$.
\end{itemize} 
\begin{prop}\label{5} 
With these notations, 
\begin{equation*} 
	R_{kl}=\sum_{sp\,R_{kl}}
	\tbinom{i}{r}D^{i-r}\La_{k-r}\cdot\p_{u_i}\La_l,
	\quad (k,l)\in\supp R,
\end{equation*}
and $R_{kl}=0$ for $(k,l)\notin\supp R$. 
\end{prop}

Further, we split 
\begin{equation*} 
	Q_{kl}=R_{kl}+Q'_{kl}+Q''_{kl}, \quad (k,l)\in\supp Q, 
\end{equation*}
where $\supp Q=\supp R\cup\supp Q'\cup\supp Q''\subset\bbI^2$, 
\begin{itemize} 
	\item 
		$\supp Q'=\cup_{(i,j)\in\supp\! R}\{(k,l)\in\bbI^2\mid k\le j, \ l\ge i, \ k+l\le i+j\}$, 
	\item 
		$\supp Q''=\cup_{(i,j)\in\supp\! R}\{(k,l)\in\bbI^2\mid k\ge j, \ l\le i, \ k+l\le i+j\}$. 
\end{itemize} 
By definition, 
$R_{kl}=0$ for $(k,l)\notin\supp R$,  
$Q^\sharp_{kl}=0$ for $(k,l)\notin\supp Q^\sharp$, where $\sharp=',''$.

\subsection{The example.}
Let us consider the operator
{$\La(D)=\la D+\frac12D\la$, $\la=\la(u_0,\dots,u_n)\in\cF$, $n\in\Z_+$.}
Here, $s=1$, $\La_0=\frac12D\la$, $\La_1=\la$, $n(i)=n+1-i$ for $i=0,1$, 
hence, $n(i)+i=n+1$. 
The operator $\La(D)$ is the Lagrange skew-adjoint by construction. 
According to Theorem \ref{T1}, the function $\la$ is defined by the system 
of differential equations $Q_{kl}=Q_{kl}(\la)=0$, $k,l\in\bbI$.  

As one can calculate, $\supp R=\{(k,l)\in\bbI^2\mid l\le1, k\le n(l)+1\}$; 

Let us denote $\la_{u_i}=\p_{u_i}\la$,  $i=0,\dots n$, 
and set 
\begin{equation*} 
	\tbinom{i}{k}'\!=\tbinom{i}{k}\!+2\tbinom{i}{k-1}, \ 
	\tbinom{i}{pkl}'=\tbinom{i}{p,k-1,l}-\tbinom{i}{p,k,l-1}=\tfrac{(k-l)i!}{p!k!l!}, 	
	 \  p+k+l=i+1.
\end{equation*} 
Now 
(remind, we assume the summation over the repeated upper and lower indices in the natural limits), 
\begin{align*} 
	R_{kl}&=\tbinom{i}{r}D^{i-r}\La_{k-r}\cdot\p_{u_i}\La_l=(k-r=0,1)= \\
		     &=\big(\tbinom{i}{k}D^{i-k}\La_0
		     	 +\tbinom{i}{k-1}D^{i+1-k}\La_1\big)\cdot\p_{u_i}\La_l= \\
		     &=\tfrac12\big(\tbinom{i}{k}D^{i+1-k}\la
		     	+2\tbinom{i}{k-1}D^{i+1-k}\la\big)\cdot\p_{u_i}\La_l
		      =\tfrac12\tbinom{i}{k}'D^{i+1-k}\la\cdot\p_{u_i}\La_l.	
\end{align*}
In particular,  
\begin{equation}\label{1} 
	R_{k1}=\tfrac12\tbinom{i}{k}'D^{i+1-k}\la\cdot\la_{u_i}, 
	\quad k=0,\dots,n+1, 
\end{equation}
while 
\begin{align*} 
	R_{k0}&=\tfrac14\tbinom{i}{k}'D^{i+1-k}\la\cdot\p_{u_i}(D\la)
					=(\p_{u_i}\com D=D\com\p_{u_i}+\p_{u_{i-1}})= \\ 
		      &=\tfrac14\tbinom{i}{k}'D^{i+1-k}\la\cdot D\la_{u_i} 
		             	  +\tbinom{i}{k}'D^{i+1-k}\la\cdot\la_{u_{i-1}}\!
		             	  =\big(\tbinom{i+1}{k}'\!-\tbinom{i}{k}'=\tbinom{i}{k-1}'\big)= \\ 
                   &=\tfrac14\big(\tbinom{i}{k}'(D^{i+2-k}\la\cdot\la_{u_i}
                                   +D^{i+1-k}\la\cdot D\la_{u_i})
                                   +\tbinom{i}{k-1}'D^{i+2-k)}\la\cdot\la_{u_i}\big)= \\
                    &=\tfrac14\big(D\big(\tbinom{i}{k}'D^{i+1-k}\la\cdot\la_{u_i}\big)
                                   +\tbinom{i}{k-1}'D^{i+1-(k-1)}\la\cdot\la_{u_i}\big),  
\end{align*}
that is 
\begin{equation}\label{2} 
	R_{k0}=\tfrac12\big(DR_{k1}+R_{k-1,1}\big), \quad k=0,\dots,n+2.
\end{equation}
Further, 
\begin{align*} 
	Q'_{kl}&=(-1)^j\tbinom{j}{p,k,l-i}D^p R_{ij}=(j=0,1)
	                =\tbinom{0}{p,k,l-i}D^p R_{i0}-\tbinom{1}{p,k,l-i}D^p R_{i1}= \\ 
                    &=\de^0_k R_{l0}-\de^0_k DR_{l1}-\de^1_{k+(l-i)}R_{i1}
                       =(DR_{l1}+R_{l-1,1}=2R_{l0})= \\ 
                    &=-(\de^0_k R_{l0}+\de^1_k R_{l1})=-R_{lk}, 
\end{align*}
that is 
\begin{equation}\label{3} 
	Q'_{kl}=-R_{lk}, \quad l=0,\dots,n+2, \quad k=0,1. 
\end{equation}
Still further, 
\begin{align*} 
	Q''_{kl}&=(-1)^i\tbinom{i}{p,k-j,l}D^p R_{ij}=(j=0,1)= \\
	              &=(-1)^i\tbinom{i}{pkl}D^p R_{i0}
	              +(-1)^i\tbinom{i}{p,k-1,l}D^p R_{i1}=(2R_{i0}=DR_{i1}+R_{i-1,1})= \\
	              &=\tfrac12(-1)^i\tbinom{i}{pkl}D^p(DR_{i1}+R_{i-1,1})
	                +(-1)^i\tbinom{i}{p,k-1,l}D^p R_{i1}= \\
                     &=\tfrac12(-1)^i\big(\tbinom{i}{p-1,k,l}-\tbinom{i+1}{pkl}
                        +2\tbinom{i}{p,k-1,l}\big)D^p R_{i1}= \\ 
                     &=\tfrac12(-1)^i\big(\tbinom{i}{p,k-1,l}-\tbinom{i}{p,k,l-1}\big)
                        D^p R_{i1}, 
\end{align*}
where we used the equality 
\begin{equation*} 
	\tbinom{i}{p-1,k,l}+\tbinom{i}{p,k-1,l}+\tbinom{i}{p,k,l-1}=\tbinom{i+1}{pkl}, 
	\quad\text{remind}\quad p+k+l=i+1.
\end{equation*}
Thus, 
\begin{equation}\label{4} 
	Q''_{kl}=\tfrac{(-1)^i}2\tbinom{i}{pkl}' D^p R_{i1}, 
	                  \quad i=0,\dots,n+1, \quad p+k+l=i+1\le n+2.
\end{equation}
Note, $Q''_{kl}=- \, Q''_{lk}$ due to the equality 
$\tbinom{i}{pkl}' =-\,\tbinom{i}{plk}' $ (see the definition).

At last, 
\begin{equation}\label{5} 
	Q_{kl}=Q_{kl}(\la)=R_{kl}-R_{lk}+\tfrac{(-1)^i}2\tbinom{i}{pkl}' 
	D^p R_{i1}=-\,Q_{lk}.
\end{equation}

Our aim now is to find all functions $\la=\la(u_0,\dots,u_n)\in\cF$, $n\in\bbZ_+$, 
such that $Q(\la)=(Q_{kl}(\la))=0$. 
It is convenient to start with the equation $Q_{1,n+1}=0$, $n\in\bbN$. 
Here, $R_{1,n+1}=0$, $R_{n+1,1}=\la\cdot\la_{u_n}$, while 
\begin{equation*} 
	Q''_{1,n+1}=\tfrac{(-1)^i}2\tbinom{i}{p,1,n+1}'D^p R_{i1}
	                   =\tfrac{(-1)^{n+1}}2\tbinom{n+1}{0,1,n+1}'R_{n+1,1}
		             =\tfrac{(-1)^n\, n}2 R_{n+1,1}.
\end{equation*}
Thus, 
\begin{equation*} 
	Q_{1,n+1}=-R_{n+1,1}+\tfrac{(-1)^n n}2 R_{n+1,1}
		=\big(\tfrac{(-1)^n\, n}2-1\big)R_{n+1,1}
		\begin{cases}=0, n=2,\\ \ne 0, n>2\end{cases}.
 \end{equation*} 
Hence, the equation $Q_{1,n+1}=0$ implies $\p_{u_n}\la=0$ for $n>2$, 
in other words,  $\la=\la(u_0,u_1,u_2)$. 
The case $n=2$ needs further detailed research. 
Namely, here, 
\begin{align*} 
	&Q_{04}=0-R_{40}+\tfrac{(-1)^2}2 R_{31}=-\tfrac12 R_{31}
		+\tfrac{(-1)^2}2 R_{31}=0, \\ 
	&Q_{03}=0-R_{30}-\tfrac12 R_{21}+\tfrac32 DR_{31}=DR_{31}-R_{21}, \\
	&Q_{02}=0-R_{20}+\tfrac12\big(3D^2 R_{31}-2DR_{21}+R_{11}\big)
		=\tfrac32D\big(DR_{31}-R_{21}\big)=\tfrac32DQ_{03}, \\
	&Q_{01}=R_{01}-R_{10}+\tfrac12\big(D^3R_{31}-D^2R_{21}+DR_{11}-
		R_{01}\big)=\tfrac12 D^2Q_{03}, \\
	&Q_{13}=0-R_{31}+R_{31}=0, \\
	&Q_{12}=0-R_{21}+\big(\tfrac32 DR_{31}-\tfrac12 R_{21}\big)
		=\tfrac32Q_{03}, 
\end{align*}
where we used the formulas (\ref{2}) and (\ref{5}).
All other equations are non relevant. Thus, the only relevant equation is 
the following: $Q_{03}=D R_{31}-R_{21}=0$. 
Here, 
\begin{align*} 
	R_{31}&=\la\cdot\la_{u_2}, \quad 
	   	                DR_{31}=D\la\cdot\la_{u_2}+\la\cdot D\la_{u_2}, 
	   	                \quad R_{21}=\la\cdot\la_{u_1}+\tfrac52 D\la\cdot \la_{u_2}, \\ 
	Q_{03}&=\la\cdot D\la_{u_2}-\tfrac32 D\la\cdot\la_{u_2}-\la\cdot\la_{u_1}
		         =u_3\big(\la\cdot\la_{u_2u_2}-\tfrac32\la^2_{u_2}\big)
		         +\tilde{Q}_{03}(u_0,u_1,u_2). 
\end{align*} 
The equation $\la\cdot\la_{u_2u_2}-\tfrac32\la^2_{u_2}=0$ has the general solution 
$\la=(\phi u_2+\psi)^{-2}$, where the functions $\phi,\psi(u_0,u_1)\in\cF$ are arbitrary. 
The reduced equation has the form
\begin{equation*} 
	\tilde{Q}_{03}=\la\cdot(\la_{u_1u_2}u_2+\la_{u_0u_2}u_1)
		-\tfrac32(\la_{u_1}u_2+\la_{u_0}u_1)\cdot\la_{u_2}-\la\cdot\la_{u_1}=0, 
\end{equation*}
where 
\begin{align*}
	&\la_{u_i}=-2(\phi u_2+\psi)^{-3}(\phi_{u_i}u_2+\psi_{u_i}), \ i=0,1, 
		\quad \la_{u_2}=-2(\phi u_2+\psi)^{-3}\phi, \\ 
	&\la_{u_iu_2}=2(\phi u_2+\psi)^{-4}\big(3(\phi_{u_i}u_2+\psi_{u_i})\phi
		-(\phi u_2+\psi)\phi_{u_i}\big), \ i=0,1.
\end{align*}
After the substitution and some simple computations we get the final form of this equation 
\begin{equation*} 
	\tilde{Q}_{03}=2(\phi u_2+\psi)^{-5}(\psi_{u_1}-\phi_{u_0}u_1)=0. 
\end{equation*}
The resulting equation $\psi_{u_1}-\phi_{u_0}u_1=0$ has the general solution 
\begin{equation}\label{6} 
	\phi(u_0,u_1)=\chi_v(u,v)\big|_{u=u_0,v=u_1^2/2}, \quad 
	\psi(u_0,u_1)=\chi_u(u,v)\big|_{u=u_0,v=u_1^2/2},
\end{equation}
where $\chi(u,v)\in\cC^\infty(\bbR^2)$ is an arbitrary smooth function. 

We summarize our calculations as the following theorem.
\begin{theorem}[cf. \cite{GD},\cite{DN}] 
The differential operator $\La(D)=\la D+\tfrac12 D\la$, $\la\in\cF$, is Hamiltonian 
iff $\la=(\phi u_2+\psi)^{-2}$, where the functions $\phi,\psi(u_0,u_1)\in\cF$ 
are defined by the formulas (\ref{6}). 
\end{theorem}

{\bf Acknowledgments.} 
This work is supported by the Russian Science Foundation under grant 14-50-00005 
and performed in Steklov Mathematical Institute of Russian Academy of Sciences.


\newpage

\begin{thebibliography}{99} 

\bibitem{Z1} 
	V.V. Zharinov, Lie-Poisson structures over differential algebras, 
	Theoret. and Math. Phys., 192:3 (2017), 1337--1349.
	
\bibitem {PJO} 
	P.J. Olver, 
	Applications of Lie groups to differential equations, 
	Springer-Verlag, New York, 1986, 1993. 

\bibitem{Z2} 
	V.V. Zharinov, Lecture notes on geometrical aspects of partial differential equations, 
	World Scientific Publishing Co. Inc., River Edge, NJ, 1992. 
	
\bibitem{I}
	N.H. Ibragimov, Transformation groups applied to mathematical physics, 
	Reidel, Boston, 1985. 
	
\bibitem{Z3} 
	V. V. Zharinov, Evolution systems with constraints in the form 
	of zero-divergence  conditions, Theoret. and Math. Phys., 163:1 (2010), 401--413.  
	
\bibitem{Z4} 
	V. V. Zharinov, B\"aclund transformations, Theoret. and Math. Phys., 
	189:3 (2016), 1681–1692. 
	
\bibitem{GD} 
	I.M. Gel'fand, I.Ya. Dorfman, Hamiltonian operators and algebraic structures 
	related to them, Func. Anal. Appl., 13 (1979), 248--262. 
	
\bibitem{DN} 
	B.A. Dubrovin, S.P. Novikov, Hydrodynamics of weakly deformed soliton lattices. 
	Differential geometry and Hamiltonian theory, Russian Math. Surveys, 44:6 (1989), 
	35--124. 
	
\bibitem{Ka}
	M. O. Katanaev, Killing vector fields and a homogeneous isotropic universe,
	 Phys. Usp., 59:7 (2016), 689--700. 

\bibitem{A} 
	G. A. Alekseev, Collision of strong gravitational and electromagnetic waves 
	in the expanding universe, Phys. Rev. D, 93:6 (2016), 061501(R).  
	
\bibitem{CS} 
	 A. P. Chugainova, V. A. Shargatov, Stability of discontinuity structures described 
	 by a generalized KdV–Burgers equation, Comput. Math. Math. Phys., 
	 56:2 (2016), 263--277. 

\bibitem{KC} 	 
	A. G. Kulikovskii, A. P. Chugainova, Study of discontinuities in solutions 
	of the Prandtl-Reuss elastoplasticity equations, 
	Comput. Math. Math. Phys., 56:4 (2016), 637--649.  
	
\bibitem{V}	
	 V. A. Vassiliev, Local Petrovskii lacunas close to parabolic singular points
	  of the wavefronts of strictly hyperbolic partial differential equations, 
	  Sb. Math., 207:10 (2016), 1363--1383.  
	  
\bibitem{Ko} 
	V. V. Kozlov, On the equations of the hydrodynamic type, 
	J. Appl. Math. Mech., 80:3 (2016), 209--214. 
	
\bibitem{Bu} 
	V. M. Buchstaber, "Polynomial dynamical systems and the 
	Korteweg--de Vries equation", Proc. Steklov Inst. Math., 294 (2016), 176--200.  
	
\bibitem{PS}
	V. P. Pavlov, V. M. Sergeev, Fluid dynamics and thermodynamics as a unified 
	field theory, Proc. Steklov Inst. Math., 294 (2016), 222--232. 

\bibitem{KC1}	
	A. G. Kulikovskii, A. P. Chugainova, A self-similar wave problem 
	in a Prandtl--Reuss elastoplastic medium, 
	Proc. Steklov Inst. Math., 295 (2016), 179--189. 

\bibitem{IC} 
	A. T. Il'ichev, A. P. Chugainova, Spectral stability theory of heteroclinic solutions 
	to the Korteweg--de Vries--Burgers equation with an arbitrary potential, 
	Proc. Steklov Inst. Math., 295 (2016), 148--157. 
	
\bibitem{A1}
	G. A. Alekseev, Integrable and non-integrable structures in 
	Einstein--Maxwell equations with Abelian isometry group G2G2, 
	Proc. Steklov Inst. Math., 295 (2016), 1--26. 	

\end{thebibliography}
\end{document}